\documentclass{article}
\usepackage{microtype}
\usepackage{graphicx}
\usepackage{caption}
\usepackage{subfigure}
\usepackage{booktabs} % for professional tables
\usepackage{hyperref}
\usepackage[accepted]{icml2020} 
\usepackage{tcolorbox}
\usepackage{amsmath,amssymb}
\usepackage{amsthm}
\usepackage{mathtools}
\usepackage{enumitem}
\usepackage{multicol}
\usepackage{bbm}
\makeatletter
\def\widebreve{\mathpalette\wide@breve}
\def\wide@breve#1#2{\sbox\z@{$#1#2$}%
     \mathop{\vbox{\m@th\ialign{##\crcr
\kern0.08em\brevefill#1{0.8\wd\z@}\crcr\noalign{\nointerlineskip}%
                    $\hss#1#2\hss$\crcr}}}\limits}
\def\brevefill#1#2{$\m@th\sbox\tw@{$#1($}%
  \hss\resizebox{#2}{\wd\tw@}{\rotatebox[origin=c]{90}{\upshape(}}\hss$}
\makeatletter

\icmltitlerunning{\name}

\newcommand{\what}{\widehat}
\newcommand{\wtilde}{\widetilde}

\newcommand{\name}{CodedFedL }
\newcommand{\Expc}{\mathbb{E}}
\newcommand{\Prob}{\mathbb{P}}

\newcommand{\bG}{\mathbf{G}}
\newcommand{\bW}{\mathbf{W}}
\newcommand{\bw}{\mathbf{w}}
\newcommand{\bg}{\mathbf{g}}
\newcommand{\bx}{\mathbf{x}}

\newcommand{\bl}{\boldsymbol{\ell}}
\newcommand{\by}{\mathbf{y}}
\newcommand{\bY}{\mathbf{Y}}
\newcommand{\bX}{\mathbf{X}}
\newcommand{\bom}{\boldsymbol\omega}
\newcommand{\bbeta}{\boldsymbol\beta}

\newcommand{\argmin}{\operatornamewithlimits{arg\,min}}
\DeclareMathOperator*{\argmax}{arg\,max}
\newcommand{\norm}[1]{\left\lVert#1\right\rVert}

\newtheorem*{theorem*}{Theorem}

\newtheorem{remark}{Remark}

\begin{document}
\twocolumn[
\icmltitle{Coded Computing for Federated Learning at the Edge}
\begin{icmlauthorlist}
\icmlauthor{Saurav Prakash}{usc}
\icmlauthor{Sagar Dhakal}{intel}
\icmlauthor{Mustafa Akdeniz}{intel}
\icmlauthor{A. Salman Avestimehr}{usc}
\icmlauthor{Nageen Himayat}{intel}
\end{icmlauthorlist}
\icmlaffiliation{usc}{Department of Electrical and Computer Engineering, University of Southern California, Los Angeles, CA}
\icmlaffiliation{intel}{Intel Labs, Santa Clara, CA}
\icmlcorrespondingauthor{Saurav Prakash}{sauravpr@usc.edu}
%\icmlcorrespondingauthor{}{}
\icmlkeywords{Machine Learning, ICML}
\vskip 0.3in
]
\printAffiliationsAndNotice{This work was part of Saurav Prakash's internship projects at Intel.\\}  % leave blank if no need to mention equal contribution

\begin{abstract}
\textit{Federated Learning} (FL) is an exciting new paradigm that enables training a global model from data generated locally at the client nodes, without moving client data to a centralized server. Performance of FL in a multi-access edge computing (MEC) network suffers from slow convergence due to heterogeneity and stochastic fluctuations in compute power and communication link qualities across clients. A recent work, Coded Federated Learning (CFL), proposes to mitigate stragglers and speed up training for \textit{linear regression tasks} by assigning \textit{redundant} computations at the MEC server. Coding redundancy in CFL is computed by exploiting statistical properties of compute and communication delays. We develop CodedFedL that addresses the difficult task of extending CFL to \textit{distributed non-linear regression} and \textit{classification} problems with multi-output labels. The key innovation of our work is to exploit distributed kernel embedding using random Fourier features that transforms the training task into distributed linear regression. We provide an analytical solution for load allocation, and demonstrate significant performance gains for CodedFedL through experiments over benchmark datasets using practical network parameters.
\end{abstract}

\section{Introduction}
We live in an era where massive amounts of data are generated each day by the IoT comprising billions of devices \cite{saravanan2019role}. To utilize this sea of distributed data for learning powerful statistical models without compromising data privacy, an exciting paradigm called federated learning (FL) \cite{konevcny2016federated,mcmahan2017communication,kairouz2019advances} has been rising recently. The  FL framework  comprises of two major steps. First, every client carries out a local update on its dataset. Second, a central server collects and aggregates the gradient updates to obtain a new model and transmits the updated model to the clients. This iterative procedure is carried out until convergence. 

We focus on the implementation of FL in multi-access edge computing (MEC) platforms \cite{guo2018efficient, shahzadi2017multi, ndikumana2019joint, ai2018edge} that enable low-latency, efficient and cloud like computing capabilities close to the client traffic. Furthermore, with the emergence of ultra-dense networks \cite{ge20165g, an2017achieving, andreev2019future}, it is increasingly likely that message transmissions during distributed learning take place over wireless. Thus, carrying out FL over MEC suffers from some fundamental bottlenecks. First, due to the heterogeneity of compute and communication resources across clients, the overall gradient aggregation at the server can be significantly delayed by the straggling computations and communication links. Second, FL suffers from wireless link failures during transmission. Re-transmission of messages can be done for failed communications, but it may drastically prolong training time. Third, data is typically non-IID across client nodes in an FL setting, i.e. data stored locally on a device does not represent the population distribution~\cite{zhao2018federated}. Thus, missing out updates from some clients can lead to poor convergence. 

In a recent work~\cite{dhakal2019}, a novel technique Coded Federated Learning (CFL) based on coding theoretic ideas, was proposed to alleviate the aforementioned bottlenecks in FL for \textit{linear regression} tasks. In CFL, at the beginning of the training procedure, each client generates masked parity data by taking linear combinations of features and labels in the local dataset, and shares it with the central server. The masking coefficients are never shared with the server, thus keeping the raw data private. During training, server performs redundant gradient computations on the aggregate parity data to compensate for the erased or delayed parameter updates from the straggling clients. The combination of coded gradient computed at the server and the gradients from the non-straggling clients \textit{stochastically} approximates the full gradient over the entire dataset available at the clients.  
    
\textbf{Our Contributions}: CFL is limited to linear regression tasks with scalar output labels. Additionally, CFL lacks theoretical analysis and evaluations over real-world datasets. \textit{In this paper, we build upon CFL and develop CodedFedL to address the difficult problem of injecting structured redundancy for straggler mitigation in FL for general non-linear regression and classification tasks with vector labels.} The key idea of our approach is to perform the kernel Fourier feature mapping~\cite{rahimi2008random} of the client data, transforming the distributed learning task into linear regression. This allows us to leverage the framework of CFL for generating parity data privately at the clients for straggler mitigation during training. For creating the parity data, we propose to take linear combinations over random features and vector labels (e.g. one-hot encoded labels for classification). Furthermore, for a given coded redundancy, we provide an analytical approach for finding the optimal load allocation policy for straggler mitigation, as opposed to the numerical approach proposed in CFL. The analysis reveals that the key subproblem of the underlying optimization can be formulated as a piece-wise concave problem with bounded domain, which can be solved using standard convex optimization tools. Lastly, we evaluate the performance of CodedFedL over two real-world datasets, MNIST and Fashion MNIST, for which CodedFedL achieves significant gains of $2.7\times$ and $2.37\times$ respectively over the uncoded approach for a small coding redundancy of $10\%$.

\textbf{Related Works}: \textit{Coded computing} is a new paradigm that has been developed for injecting computation redundancy in unorthodox encoded forms to efficiently deal with communication bottleneck and system disturbances like stragglers, outages, node failures, and adversarial computations in distributed systems~\cite{li2017fundamental,lee2017speeding, tandon2017gradient,karakus2017straggler,reisizadeh2019coded}. Particularly, ~\cite{lee2017speeding} proposed to use erasure coding for speeding up distributed matrix multiplication and linear regression tasks. \cite{tandon2017gradient} proposed a coding method over gradients for synchronous gradient descent. ~\cite{karakus2019} proposed to encode over the data for avoiding the impact of stragglers for linear regression tasks. Many other works on coded computing for straggler mitigation in distributed learning have been proposed recently \cite{ye2018communication,dhakal2019coded,yu2017polynomial,raviv2017gradient,charles2017approximate}. In all these works, the data placement and coding strategy is orchestrated by a central server. As a result, these works are not applicable in FL setting, where data is privately owned by clients and cannot be shared with the central server. The CFL paper \cite{dhakal2019} was the first to propose a distributed method for coding in a federated setting, but was limited to linear regression. 

Our work, CodedFedL, proposes to develop coded parity data for non-linear regression and classification tasks with vector labels. For this, we leverage the popular kernel embedding based on random Fourier features (RFF)~\cite{rahimi2008random,pham2013fast,kar2012random}, that has been a popular approach for dealing with kernel-based inference on large datasets. In the classical kernel approach~\cite{shawe2004kernel}, a kernel similarity matrix needs to be constructed, whose storage and compute costs are quadratic in the size of the dataset. RFF was proposed in \cite{rahimi2008random} to address this problem by explicitly constructing finite-dimensional random features from the data such that inner products between the random features approximate the kernel functions. Training and inference with random features have been shown to work considerably well in practice~\cite{rahimi2008random,rahimi2008uniform,shahrampour2018data, kar2012random, rick2016random}. 
\section{Problem Setup}\label{sec:background}

In this section, we provide a preliminary background on linear regression and FL, followed by a description of our computation and communication models.
\subsection{Preliminaries}
\label{sec:prelim}
Consider a dataset $D{=}\{(\bx_1,\by_1),\ldots,(\bx_m,\by_m)\}$, where for $i{\in}\{1,\ldots,m\}$, data feature $\bx_i{\in}\mathbb{R}^{1\times d}$ and label $\by_i {\in}\mathbb{R}^{1\times c}$. Many machine learning tasks consider the following optimization problem:
\begin{align}
\label{eq:mlopt}
\bbeta^{*}&=\argmin_{\bbeta\in \mathbb{R}^{d\times c}} \frac{1}{2}\norm{\bX\bbeta-\bY}_F^2+\frac{\lambda}{2}\norm{\bbeta}_F^2,\nonumber\\
&=\argmin_{\bbeta\in \mathbb{R}^{d\times c}} \frac{1}{2m}\sum_{i=1}^m\norm{\bx_i\bbeta-\by_i}_2^2+\frac{\lambda}{2}\norm{\bbeta}_F^2,
\end{align}
where $\bX{\in} \mathbb{R}^{m\times d}$ and $\bY{\in} \mathbb{R}^{m\times c}$ denote the feature matrix and label matrix respectively, while $\bbeta$, $\lambda$ and $\norm{\cdot}_F$ denote the model parameter, regularization parameter and Frobenius norm respectively. A common strategy to solve (\ref{eq:mlopt}) is gradient descent. Specifically, model is updated sequentially until convergence as follows: $\bbeta^{(r+1)} = \bbeta^{(r)} {-} {\mu^{(r+1)}}\left({\bg}+{\lambda} \bbeta^{(r)}\right)$, where ${\mu^{(r+1)}}$ denotes the learning rate. Here, $\bg$ denotes the gradient of the loss function over the dataset $D$ as follows: $\bg=\frac{1}{m}\bX^T(\bX\bbeta^{(r)}-\bY).$
%In federated learning, dataset $D$ is distributedly available at the clients, as we describe next. 

In FL, the goal is to solve (\ref{eq:mlopt}) for the dataset $D=\cup_{j=i}^n D_j$, where $D_j$ is the dataset that is available locally at client $j$. Let $\ell_j{>}0$ denote the size of $D_j$, and $\bX^{(j)} {=} [\bx^{(j)T}_1, \dots, \bx^{(j)T}_{\ell_j}]^T$ and $\bY^{(j)} {=} [\by^{(j)T}_1, \dots, \by^{(j)T}_{\ell_j}]^T$ denote the feature set and label set respectively for the $j$-th client. Therefore, the combined feature and label sets across all clients can be represented as $\bX = [\bX^{(1)T}, \dots, \bX^{(n)T}]^T$ and $\bY = [\bY^{(1)T}, \dots, \bY^{(n)T}]^T$ respectively. 

During iteration ${(r+1)}$ of training, the server shares the current model $\bbeta^{(r)}$ with the clients. Client $j$ then computes the local gradient $\bg^{(j)}{=} \frac{1}{\ell_j}\bX^{(j)T}(\bX^{(j)}\bbeta^{(r)} - \bY^{(j)})$. The server collects the client gradient and combines them to recover the full gradient $\bg = \frac{1}{m}\sum_{j=1}^n \ell_j\bg^{(j)} = \frac{1}{m}\bX^T(\bX\bbeta^{(r)}-\bY)$. The server then carries out the model update to obtain $\bbeta^{(r+1)}$, which is then shared with the clients in the following training iteration. The iterative procedure is carried out until sufficient convergence is achieved.

To capture the stochastic nature of implementing FL in MEC, we consider probabilistic models for computation and communication resources as illustrated next.

\subsection{Models for Computation and Communication}
\label{subsec:delayModel}
To statistically represent the compute heterogeneity, we assume a shifted exponential model for local gradient computation. Specifically, the computation time for $j$-th client is given by the  shifted exponential random variable $T_{cmp}^{(j)}$ as $T_{cmp}^{(j)} {=} T_{cmp}^{(j,1)} {+} T_{cmp}^{(j,2)}.$ Here, $T_{cmp}^{(j,1)}{=} \frac{\ell_j}{\mu_j}$ represents the non-stochastic part of the time in seconds to process partial gradient over $\ell_j$ data points, where processing rate is $\mu_j$ data points per second. $T_{cmp}^{(j,2)}$ models the stochastic component of compute time coming from random memory access during read/write cycles associated to Multiply-Accumulate (MAC) operations, where $p_{T_{cmp}^{(j,2)}}(t) {=} \gamma_j e^{-\gamma_j t},\,t{\geq} 0$. Here, $\gamma_j {=}  \frac{\alpha_j\mu_j}{\ell_j}$, with $\alpha_j{>}0$ controlling the average time spent in computing vs memory access.

The overall execution time for $j$-th client during $(r+1)$-th epoch also includes  $T_{com-d}^{(j)}$, time to download $\bbeta^{(r)}$ from the server, and $T_{com-u}^{(j)}$, time to upload the partial gradient $\bg^{(j)}$ to the server. The communications between server and clients take place over wireless links, that fluctuate in quality. It is a typical practice to model the wireless link between server and $j$-th client by a tuple $(r_j,p_j)$, where $r_j$ and $p_j$ denote the the achievable data rate (in bits per second per Hz) and link erasure probability $p_j$~\cite{3gpp}. Downlink and uplink communication delays are IID random variables given as follows\footnote{For the purpose of this article, we assume the downlink and uplink delays to be reciprocal. Generalization of our framework to asymmetric delay model is easy to address.}: $T_{com-d}^{(j)}{=}N_j \tau_j,\,T_{com-u}^{(j)}{=}N_j \tau_j$. Here, $\tau_j{=}\frac{b}{r_j W}$ is the deterministic time to upload (or download) a packet of size $b$ bits containing partial gradient $\bg^{(j)}$ (or model $\bbeta^{(r)}$) and $W$ is the bandwidth in Hz assigned to the $j$-th worker device. Here, $N_j{\sim}G(p=1-p_j)$ is a geometric random variable denoting the number of transmissions required for the first successful communication as follows:
\begin{align}
\label{eq:commModel}
\Prob\{N_j=x\} = p^{x-1}_j(1-p_j), \,\, x = 1,2,3,...
\end{align}
Therefore, the total time $T^{(j)}$ taken by the $j$-th device to receive the updated model, compute and successfully communicate the partial gradient to the server is as follows:
\begin{align}
\label{eq:totDelay}
T^{(j)} = T_{comp}^{(j)} + T_{comm-d}^{(j)} + T_{comm-u}^{(j)},
\end{align}
with average delay $
E(T^{(j)}) = \frac{\ell_j}{\mu_j} (1 +\frac{1}{\alpha_j}) + \frac{2 \tau_j}{1-p_j}$. 
\section{Proposed  \name   Scheme}\label{sec:codedfedl}
\label{sec:cgd}
In this section, we present the different modules of our proposed \name scheme for resilient FL in MEC networks: random Fourier feature mapping for non-linear regression, structured optimal redundancy for mitigating stragglers, and modified training at the server.  

\subsection{Kernel Embedding}
\label{sec:kernelembed}
Linear regression procedure outlined in Section \ref{sec:prelim} is computationally favourable for low-powered personalized client devices,  as the gradient computations involve matrix multiplications which have low computational complexity. However, in many ML problems, a linear model does not perform well. To combine the advantages of non-linear models and low complexity gradient computations in linear regression in FL, we propose to leverage kernel embedding based on random Fourier feature mapping (RFFM)~\cite{rahimi2008random}. In RFFM, each feature $\bx_i{\in} \mathbb{R}^{1\times d}$ is mapped to $\widehat{\bx}_i$ using a function $\phi:\mathbb{R}^{1\times d}\rightarrow \mathbb{R}^{1\times q}$. RFFM approximates a positive definitive kernel function $K:\mathbb{R}^{1\times d}\times \mathbb{R}^{1\times d}\rightarrow \mathbb{R}^{1\times q}$ as represented below:
\begin{align}
    K(\bx_{i},\bx_{j}) \approx \what{\bx}_i\what{\bx}^T_j=\phi(\bx_{i})\phi(\bx_{j})^T.
\end{align}
Before training starts, $j$-th client carries out RFFM to transform its raw feature set $\bX^{(j)}$ to $\what{\bX}^{(j)}{=}\phi(\bX^{(j)})$, and training proceeds with the transformed dataset $\what{D}{=}(\what{\bX},\bY)$, where $\what{\bX}{\in} \mathbb{R}^{m\times q}$ is the matrix denoting all the transformed features across all clients. The goal then is to recover the full gradient $\what{\bg}{=}\frac{1}{m}\what{\bX}^T(\what{\bX}\bbeta^{(r)}-\what{\bY})$ over $\what{D}$ for $\bbeta^{(r)}{\in}\mathbb{R}^{q\times c}$. In this paper, we consider a commonly used kernel known as RBF kernel~\cite{vert2004primer}
$K(\bx_i,\bx_j) = \exp\left(-\frac{||\bx_i-\bx_j||^2}{2\sigma^2}\right),$
where $\sigma$ is the kernel width parameter. RFFM for the RBF kernel is obtained as follows~\cite{rahimi2008uniform}:
\begin{align}
\label{eq:rffm}
    \what{\bx}_i=\sqrt{\frac{2}{q}}\left[\cos(\bx_i\bom_1+\delta_1),\ldots,  \cos(\bx_i\bom_q+\delta_q))\right]
\end{align}
where the frequency vectors $\bom_s\in \mathbb{R}^{d\times 1}$ are drawn independently from $\mathcal{N}(\mathbf{0}, \frac{1}{\sigma^2}\mathbf{I}_d)$, while the shift elements $\delta_s$ are drawn independently from $\text{Uniform}(0,2\pi]$ distribution.
\begin{remark}
For distributed transformation of features at the clients, the server sends the same pseudo-random seed to every client which then obtains the samples required for RFFM in (\ref{eq:rffm}). This mitigates the need for the server to communicate the samples $\bom_1,\ldots, \bom_q,\delta_1,\ldots,\delta_q$ thus minimizing client bandwidth overhead.
\end{remark}
Along with the computational benefits of linear regression over the transformed dataset $\what{D}{=}(\what{\bX},\bY)$, applying RFFM allows us to apply distributed encoding strategy for linear regression developed in \cite{dhakal2019}. The remaining part of Section \ref{sec:codedfedl} has been adapted from the CFL scheme proposed in \cite{dhakal2019}.  

\subsection{Distributed Encoding}\label{encode}
Client $j$ carries out random linear encoding over its transformed training dataset  $\what{D}_j{=}(\what{\bX}^{(j)},{\bY}^{(j)})$ containing transformed feature set $\what{X}^{(j)}$ obtained from RFFM. Specifically, random generator matrix $\bG_j{\in} \mathbb{R}^{u \times \ell_j}$ is used for encoding, where $u$ denotes the \textit{coding redundancy} which is the amount of parity data to be generated at each device. Typically, $u{\ll}m$. Further discussion on $u$ is deferred to Section \ref{sec:red}, where the load allocation policy is described. 

Entries of $\bG_j$ are drawn independently from a normal distribution with mean $0$ and variance $\frac{1}{u}$. $\bG_j$ is applied on the \textit{weighted} local dataset to obtain $\widebreve{D}_j{=}(\widebreve{\bX}^{(j)},\widebreve{\bY}^{(j)})$ as follows: $\widebreve{\bX}^{(j)}{=}\bG_j \bW_j \what{\bX}^{(j)},\,\,  \widebreve{\bY}^{(j)}{=}\bG_j \bW_j \bY^{(j)}$.
For $\bw_j=[w_{j,1},\ldots,w_{j,\ell_j}]$, the weight matrix $\bW_j=\text{diag}(\bw_j)$ is an $\ell_j{\times}\ell_j$ diagonal matrix that weighs training data point $(\what{\bf x}^{(j)}_k,\by^{(j)}_k)$ with $w_{j,k}$ based on the stochastic conditions of the compute and communication resources, $k{\in}[\ell_j]$. We defer the details of deriving $\bW_j$ to Section \ref{sec:weight}. 

The central server receives local parity data from all client devices and combines them to obtain the \textit{composite} parity dataset $\widebreve{D}{=}(\widebreve{\bX},\widebreve{\bY})$, where $\widebreve{\bX} {\in} \mathbb{R}^{u \times q}$ and $\widebreve{\bY}{\in} \mathbb{R}^{u \times c}$ are the composite parity feature set and label set as follows: $
\widebreve{\bX} {=} \sum_{j=1}^n \widebreve{\bX}^{(j)},\,\,  \widebreve{\bY} {=} \sum_{j=1}^n \widebreve{\bY}^{(j)}$. 
Therefore, we have:
\begin{align}
\label{eq:parityX}
\widebreve{\bX}= \bG \bW \what{\bX},\,
\widebreve{\bY}= \bG \bW {\bY},
\end{align}
where $\bG {=} [{\bf G}_1, \dots, {\bf G}_n]{\in} \mathbb{R}^{u\times m}$ and ${\bf W}{\in}\mathbb{R}^{m\times m}$ is a block-diagonal matrix given by 
$\bW=\text{diag}([\bw_1,\ldots,\bw_n])$. Equation (\ref{eq:parityX}) represents the encoding over \textit{the entire decentralized dataset} $\what{D}{=}(\what{\bX}, \bY)$, performed implicitly in a distributed manner across clients. 
\begin{remark}
	Although client $j{\in}[n]$ shares its locally coded dataset ${\widebreve{D}}_j{=}(\widebreve{\bX}^{(j)},\widebreve{\bY}^{(j)})$ with the central server, the local dataset $\what{D}_j$ as well as the encoding matrix $\bG_j$ are private to the client and not shared with the server. It is an interesting future work to characterize the exact privacy leakage after the proposed randomization.  
\end{remark}
\subsection{Coding Redundancy and Load Assignment} 
\label{sec:red}
The server carries out a load policy design based on statistical conditions of MEC for finding $\wtilde{\ell}_j{\leq} \ell_j$, the number of data points to be processed at the $j$-th client, and $u{\leq} u^{max}$, the number of coded data points to be processed at the server,  where, $u^{max}$ is the maximum number of coded data points that the server can process. For each epoch, let $R_j (t;\wtilde{\ell}_j)$ be the indicator random variable denoting the event that server receives the partial gradient computed and communicated by the $j$-th client within time $t$, i.e. $R_j (t;\wtilde{\ell}_j){=}\wtilde{\ell}_j {\mathbbm{1}}_{\{T_j\leq t\}}$.  Clearly, $R_j (t;\wtilde{\ell}_j){\in}\{0,\wtilde{\ell}_j\}$. The following denotes the total aggregate return for $t\geq0$:
\begin{align}
\label{eq:aggReturn}
R(t;(u,\wtilde{\mathbf{\bl}}))&=R_C(t;u)+R_U(t;\bl),
\end{align}
where $R_{C}(\cdot)$ denotes indicator random variable for the event that the server finishes computing the \textit{coded gradient} over the parity dataset $\widebreve{D}{=}(\widebreve{\bX},\widebreve{\bY})$, while $R_U(t;\bl){=}\sum_{j=1}^{n} R_j(t;\ell_j)$ denotes the total aggregate return for the \textit{uncoded} partial gradients from the clients. The goal is to have an expected return $\Expc(R(t;\wtilde{\bl})){=}m$ for a minimum waiting time $t{=}t^*$, where $m$ is the total number data points at the clients. When coding redundancy is large, clients need to compute less. This however results in a coarser approximation of the true gradient over the entire distributed client data, since encoding of training data results in colored noise, which may result in poor convergence. 

Without loss of generality, we assume that the server has reliable and powerful computational capability so that $R_{C}(t;u^{max}){=}u^{max}$ a.s. for any $t\geq 0$. Therefore, the problem reduces to finding an expected aggregate return from the clients $\Expc(R_U(t;\bl))=(m-u^{max})$ for a \textit{minimum waiting time} $t=t^*$ for each epoch. The two-step approach for the load allocation is described below:    \\
\textbf{Step 1}: Optimize $\wtilde{\bl}{=}(\wtilde{\ell}_1,\ldots,\wtilde{\ell}_n)$ to maximize the expected return $\Expc(R_U(t;\bl))$ for a fixed $t{>}0$ by solving the following for the expected aggregate return from the clients:
\begin{align}
\label{eq:optimLn}
\bl^* (t) = \argmax_{\boldsymbol{0} \leq \wtilde{\bl}\leq (\ell_1,\ldots,\ell_n)} \Expc(R_U (t;\wtilde{\bl}))
\end{align}
Also, the optimization problem in (\ref{eq:optimLn}) is decomposable into $n$ independent optimization problems, one for each client $j$: 
\begin{align}
\label{eq:optimLj}
\ell_j^* (t) = \argmax_{0 \leq \wtilde{\ell}\leq \ell_j} \Expc(R_j (t;\wtilde{\ell_j}))
\end{align}
\begin{remark}
\label{rmk:concavityExpRetClnt}
We prove in Section \ref{sec:theory} that $\Expc(R_j (t;\wtilde{\ell}_j))$ is a piece-wise concave function in $\ell_j{>}0$, where the interval boundaries are determined by the number of transmissions during the return time for $j$-th client. Thus, (\ref{eq:optimLj}) can be efficiently solved using any convex toolbox.
\end{remark}
\textbf{Step 2}: Next, optimization of $t$ is considered in order to find the minimum waiting time so that the maximized expected aggregate return $\Expc(R(t;{\bl}^{*}(t)))$ is equal to $m$, where ${\bl}^{*}(t)$ is the solution to (\ref{eq:optimLn}). Specifically, for a tolerance parameter $\epsilon\geq0$, the following optimization problem is considered:
\begin{align}
\label{eq:epochTime}
t^* = \argmin_{t \geq 0}: m-u \leq \Expc(R_U(t;{\bl}^{*}(t)))\leq m-u +\epsilon.
\end{align} 
\begin{remark}
\label{rmk:monotoneExpRetClnt}
In Section \ref{sec:theory}, we leverage our derived mathematical result for $\Expc(R_j (t;\wtilde{\ell}_j))$ to numerically show that $\Expc(R(t;{\bl}^{*}(t)))$ is monotonically increasing in $t$. Therefore, (\ref{eq:epochTime}) can be efficiently solved using a binary search for $t$. 
\end{remark}
\begin{remark}
The optimization procedure outlined above can be generalized to solving (\ref{eq:aggReturn}), by treating server as $(n+1)$-th node, with $0{\leq}\wtilde{\ell}_{n+1}{\leq} u^{max}$. Then, $u=\ell^{*}_{n+1}(t^{*})$.
\end{remark}
\subsection{Weight Matrix Construction} 
\label{sec:weight}
Client $j$ samples $\ell_j^{*}(t^{*})$ data points uniformly and randomly that it will process for local gradient computation during training. It is not revealed to the server which data points are sampled, adding another layer of privacy. The probability that the partial gradient computed at client $j$ is not received at server by $t^{*}$ is $pnr_{j,1}{=}(1-\Prob\{T_j \leq t^{*}\})$. Also, $(\ell_j-\ell_j^{*}(t^{*}))$ data points are never evaluated locally, which  implies that the probability of no return $pnr_{j,2}{=}1$. 

The diagonal weight matrix $\bW_j{\in}\mathbb{R}^{\ell_j\times \ell_j}$ captures this absence of updates reaching the server during the training procedure for different data points. Specifically, for the $\ell_j^{*}(t^{*})$ data points processed at the client, the corresponding weight matrix coefficient is $w_{j,k} {=} \sqrt{pnr_{j,1}}$, while for the $(\ell_j-\ell_j^{*}(t^{*}))$ data points never processed, $w_{j,k} {=} \sqrt{pnr_{j,2}}$. As we illustrate next, this weighing ensures that the combination of the coded gradient and the partial gradient updates from the non-straggling clients stochastically approximates the full gradient over the entire dataset across the clients.

\subsection{Coded Federated Aggregation}
\label{sec:codfedagg}
In each epoch, the server computes the coded gradient ${\bg}_C$ over the composite parity data $\widebreve{D}{=}(\widebreve{\bX},\widebreve{\bY})$ as follows:
\begin{align}
{\bg}_C&=\widebreve{\bX}^{T}(\widebreve{\bX} \bbeta^{(r)} - \widebreve{\bY})\nonumber\\
&=\what{\bf X}^T {\bf W}^T\biggl({\bf G}^T{\bf G}\biggr){\bf W}(\what{\bf X}{\boldsymbol\beta}^{(r)}-\what{\bf Y})
\end{align}
\begin{align}
\label{eq:codedGrad}
\implies\Expc({\bg}_C)&\overset{(a)}= \what{\bf X}^T {\bf W}^T{\bf W}(\what{\bf X}{\boldsymbol\beta}^{(r)}-\what{\bf Y})\nonumber\\
&= \sum_{j=1}^{n}\sum_{k=1}^{\ell_j} w_{j,k}^2 \what{\bf x}^{(j)T}_k(\what{\bf x}^{(j)}_k {\boldsymbol\beta}^{(t)} - \what{\by}^{(j)}_k)
\end{align}
In $(a)$, we have replaced the quantity $\Expc(\bG^T \bG)$ by an identity matrix since the entries in $\bG\in\mathbb{R}^{u\times m}$ are IID with variance ${1}/{u}$. One can even replace $\bG^T \bG$ by an identity matrix as a good approximation for reasonably large $u$.

Client $j$ computes gradient ${\bg}_U^{(j)}{=}\frac{1}{\ell_j^{*}(t^{*})}\wtilde{\bX}^T(\wtilde{\bX}\bbeta^{(r)}-\wtilde{\bY})$ over $\wtilde{D}^{(j)}{=}(\wtilde{\bX}^{(j)},\wtilde{\bY}^{(j)})$, which is composed of the $\ell_j^{*}(t^{*})$ data points that $j$-th client samples for processing before training. The server waits for the partial gradients from the clients until the optimized waiting time $t^*$ and aggregates them to obtain ${\bg}_U{=}\sum_{j=1}^n 
\ell_j^{*}(t^{*}){\mathbbm{1}}_{\{T_j\leq t^{*}\}} {\bg}_U^{(j)}$. The server combines the coded and uncoded gradients to  obtain $\bg_M{=}\frac{1}{m}(\bg_C+\bg_U)$, that stochastically approximates the full gradient $\what{\bg}{=}\frac{1}{m}\what{\bX}^T(\what{\bX}\bbeta^{(r)}-\what{\bY})$.
Specifically, the expected aggregate gradient from the clients is as follows:
\begin{align}
\label{eq:receivedGrad}
&\Expc({\bg}_U)=\sum_{j=1}^n\Prob(T_j\leq t^{*}) {\bg}_U^{(j)}\nonumber\\
&\overset{(a)}=\sum_{j=1}^{n}\sum_{k=1}^{\ell^{*}_j(t)}\Prob(T_j \leq t^*)\what{\bf x}^{(j)T}_k(\what{\bf x}^{(j)}_k {\boldsymbol\beta}^{(t)} - \what{\by}^{(t)}_k) \nonumber\\
&\overset{(b)}=\sum_{j=1}^{n}\sum_{k=1}^{l_j}(1-w_{j,k}^2)\what{\bf x}^{(j)T}_k(\what{\bf x}^{(j)}_k {\boldsymbol\beta}^{(t)} - \what{\by}^{(j)}_k),
\end{align}
where the inner sum in $(a)$ denotes the sum over the data points in $\wtilde{D}^{(j)}=(\wtilde{\bX}^{(j)},\wtilde{\bY}^{(j)})$, while in $(b)$ all the points in the local dataset are included, with $(1-w_{j,k}^2)=0$ for the points that are selected to be not processed by the $j$-th client. In light of (\ref{eq:codedGrad}) and (\ref{eq:receivedGrad}), we can see that $\Expc(\bg_M)=\what{\bg}$.

\section{Analyzing \name}\label{sec:theory}
In this section, we analyze the expected return $\Expc(R_j (t;\wtilde{\ell}_j))$ defined in Section \ref{sec:red}. We first present the main result.
\begin{theorem*} 
\label{thm:expcRetWorker}
For the compute and communication models described in Section \ref{subsec:delayModel}, let $0{\leq}\wtilde{\ell}_j{\leq}\ell_j$ be the number of data points processed by $j$-th client in each training epoch. Then, for a waiting time of $t$ at the server, the expectation of the return $R_j(t;\wtilde{\ell}_j){=}\wtilde{\ell}_j {\mathbbm{1}}_{\{T_j\leq t\}}$ satisfies the following: 
\begin{align}
    &\Expc({R_{j}(t;\wtilde{\ell}_j)})\nonumber\\
    &=\left\{
	\begin{array}{ll}
		\sum_{\nu=2}^{\nu_m}{U}\bigg(t-\frac{\wtilde{\ell}_j}{\mu_j} -\tau_j \nu\bigg)h_{\nu} f_{\nu}(t;\wtilde{\ell}_j)  & \mbox{if } \nu_m\geq 2   \\
		0 & \mbox{otherwise } \nonumber
	\end{array}
    \right.
\end{align}
where $U(\cdot)$ is the unit step function,\\$f_{\nu}(t;\wtilde{\ell}_j)=\wtilde{\ell}_j\bigg(1-e^{-\frac{\alpha_j\mu_j}{\wtilde{\ell}_j}(t-\frac{\wtilde{\ell}_j}{\mu_j}-\tau_j \nu)}\bigg)$,\\
$h_{\nu}= (\nu-1)(1-p_j)^2p_j^{\nu-2}$, \\
and $\nu_m\in\mathbb{Z}$ satisfies $t-\tau_j \nu_m>0,\, t-\tau_j (\nu_m+1)\leq 0$.
%\vspace{-1cm}
\end{theorem*}
The proof is in Appendix \ref{sec:append1}. Next we discuss the behavior of $\Expc(R_j (t;\wtilde{\ell}_j))$ for $\nu_m{\geq}2$. For a fixed $t{>}0$, consider the function $f_{\nu}(t;\wtilde{\ell}_j)$ for $\wtilde{\ell}_j{>}0$. Then, the following holds: 
\begin{align}
f''_{\nu}(t;\wtilde{\ell}_j) = -e^{-\frac{\alpha_j\mu_j}{\wtilde{\ell}_j}(t-\nu\tau_j-\frac{\wtilde{\ell}_j}{\mu_j})}\frac{{\alpha_j^2\mu_j}^2 (t-\nu\tau_j)^2}{\wtilde{\ell}^3_j}<0.\nonumber
\end{align}
Thus, $f_{\nu}(t;\wtilde{\ell}_j)$ is strictly concave in the domain $\wtilde{\ell}_j{>}0$. Also, $f_{\nu,t}(\wtilde{\ell}_j){\leq}0$ for $\wtilde{\ell}_j{\geq}{\mu_j(t-\tau_j \nu)}$ for $\nu\in\{2,\ldots,\nu_m\}$. Solving for $f'_{\nu}(t;\wtilde{\ell}_j)=0$, we obtain the optimal load as follows:  
\begin{align}
\label{eq:optLfn}
    \ell_{j}^*(t,\nu) = -\frac{\alpha_j \mu_j}{W_{-1}(-e^{-(1+\alpha_j)})+1}(t-\nu\tau_j)U(t-\nu\tau_j).
\end{align}
where $W_{-1}(\cdot)$ is the minor branch of Lambert $W$-function, where the Lambert $W$-function is the inverse function of $f(W)=We^W$.
\begin{figure}[h!]
  \centering
  \subfigure[Illustration of the piece-wise concavity of the function $\Expc(R_j(t;\wtilde{\ell}_j))$ for $\wtilde{\ell}_j>0$.]{\includegraphics[scale=0.3]{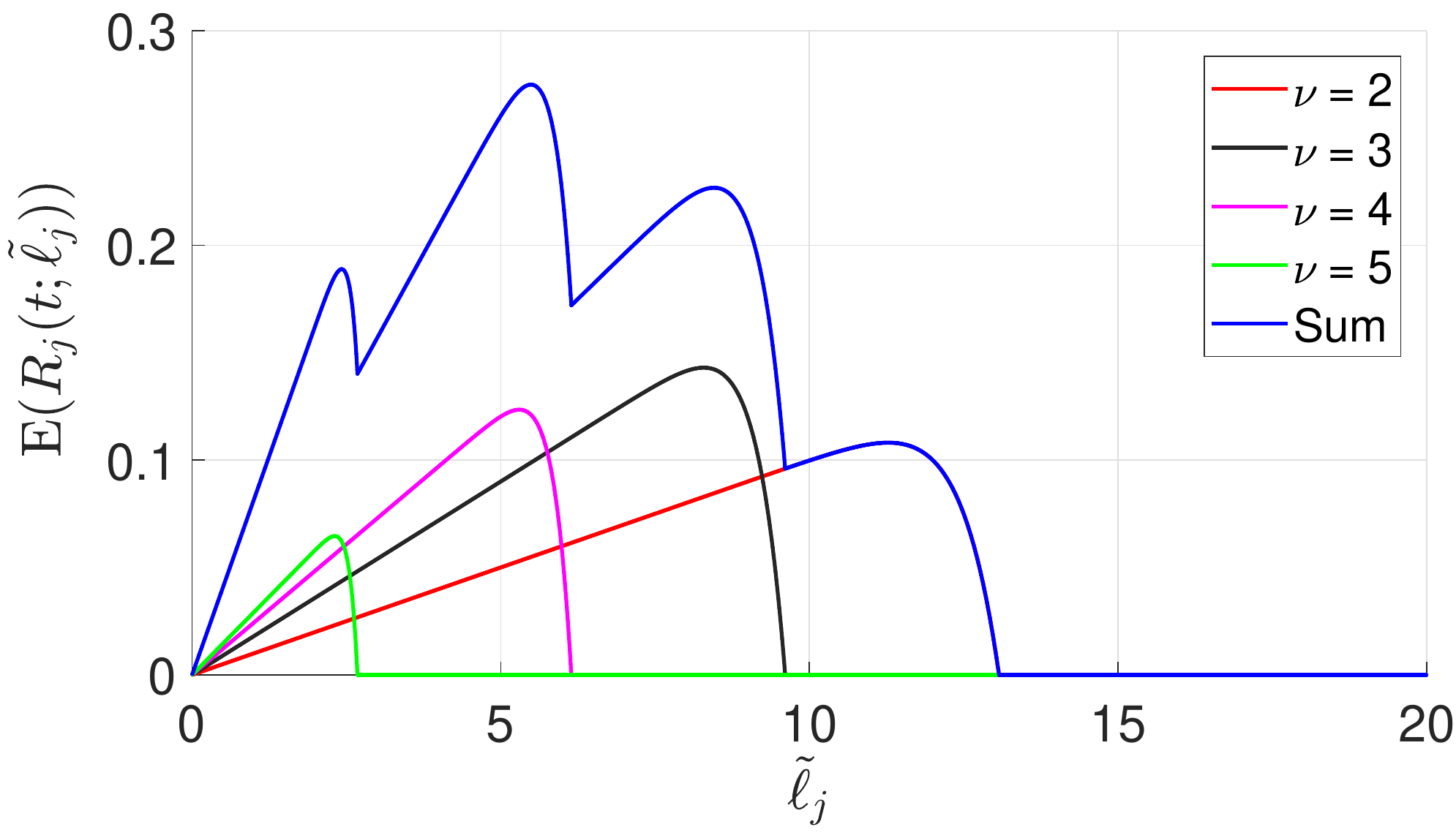}\label{fig:piecewiseConc}}\quad
  \subfigure[Illustrating the monotonic relationship between $\Expc(R_j (t;{\ell}^{*}_j(t)))$ and $t$ for $j$-th client.]{\includegraphics[scale=0.3]{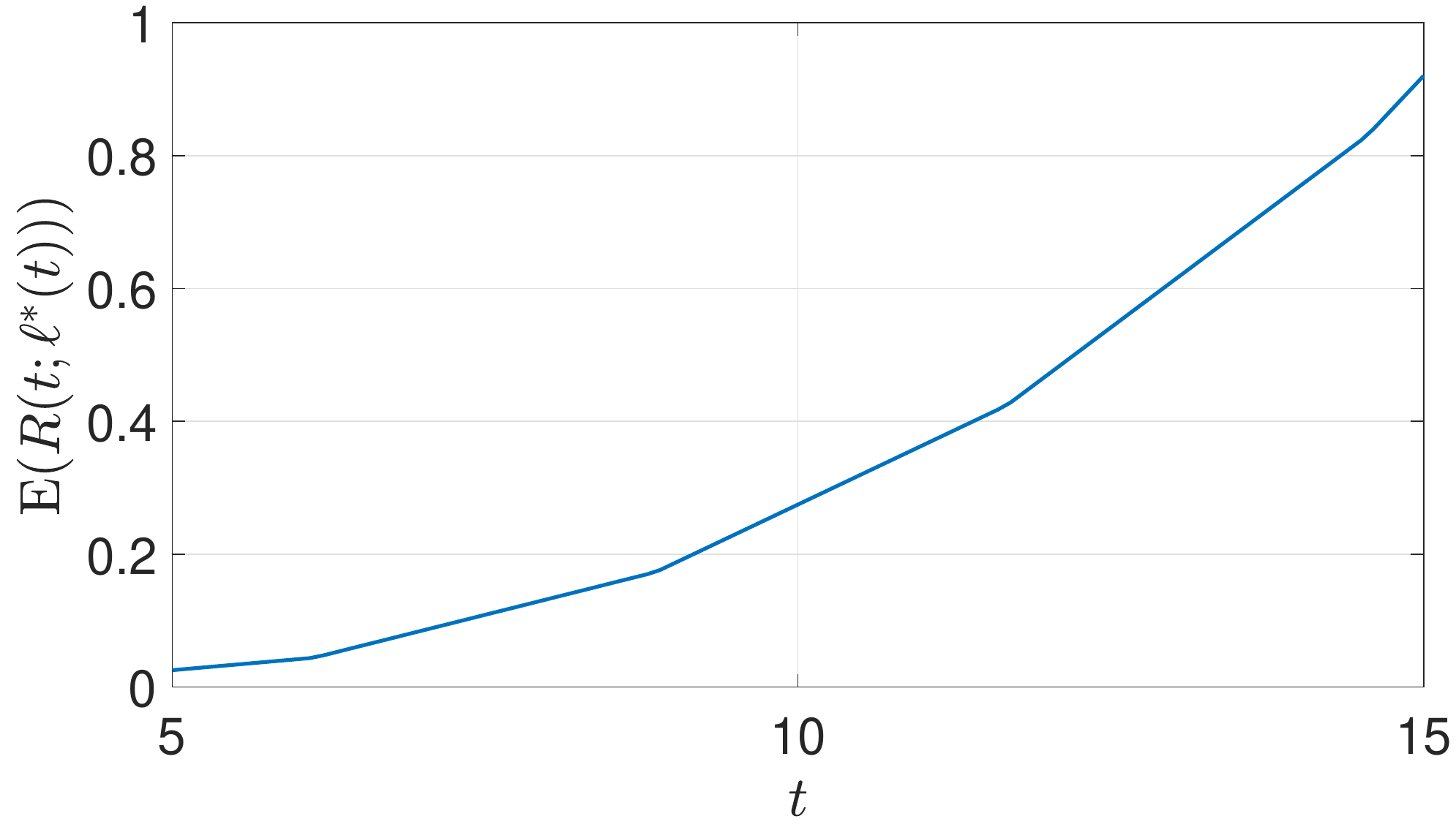}\label{fig:optretVSt}}\caption{Illustrating the properties of expected aggregate return $\Expc(R_j(t;\wtilde{\ell}_j))$ based on the result of the Theorem. We assume $p_j{=}0.9$, $\tau_j{=}\sqrt{3}$, $\mu_j{=}2$ and in Fig. \ref{fig:piecewiseConc}, we have fixed $t{=}10$.}\label{fig:illustrate_analysis}
 \end{figure}
 Therefore, as highlighted in Remark \ref{rmk:concavityExpRetClnt}, the expected return $\Expc(R_j (t;\wtilde{\ell}_j))$ is piece-wise concave in $\wtilde{\ell}_j$ in the intervals $(0,\mu_j(t-\nu_m\tau_j)),\ldots,(\mu_j(t-3\tau),\mu_j(t-2\tau))$. Thus, for a given $t{>}0$, the problem of maximizing the expected return decomposes into a finite number of convex optimization problems, that are efficiently solved in practice~\cite{boyd2004convex}. This piece-wise concave relationship is also highlighted in Fig. \ref{fig:piecewiseConc}.
 
Consider the optimized expected return $\Expc(R_j (t;{\ell}^{*}_j(t)))$. Intuitively, as we increase the  waiting time at the server, the optimized load allocation should vary such that the server gets more return on average. We substantiate this intuition and illustrate the relationship in Fig. \ref{fig:optretVSt}. As $\Expc(R_j (t;\wtilde{\ell}^{*}_j(t)))$ is monotonically increasing for each client $j$, it holds for $\Expc(R_U (t;{\bl}^{*}(t))){=}\sum_{j=1}^{n} \Expc(R_j (t;{\ell}^{*}_j(t)))$. Hence, the optimization problem in (\ref{eq:epochTime}) can be solved efficiently using binary search, as claimed in Remark \ref{rmk:monotoneExpRetClnt}. 

\section{Empirical Evaluation of \name}\label{sec:experiments}
We now illustrate the numerical gains of \name proposed in Section \ref{sec:codedfedl}, comparing it with the uncoded approach, where each client computes partial gradient over the entire local dataset, and the server waits to aggregate local gradients from all the clients. As a pre-processing step, kernel embedding is carried out for obtaining random features for each dataset, as outlined in Section \ref{sec:kernelembed}.
 \begin{figure}[h!]
  \centering
  \subfigure[The test accuracy with respect to the wall-clock time under uncoded scheme vs. CodedFedL scheme.]{\includegraphics[scale=0.43]{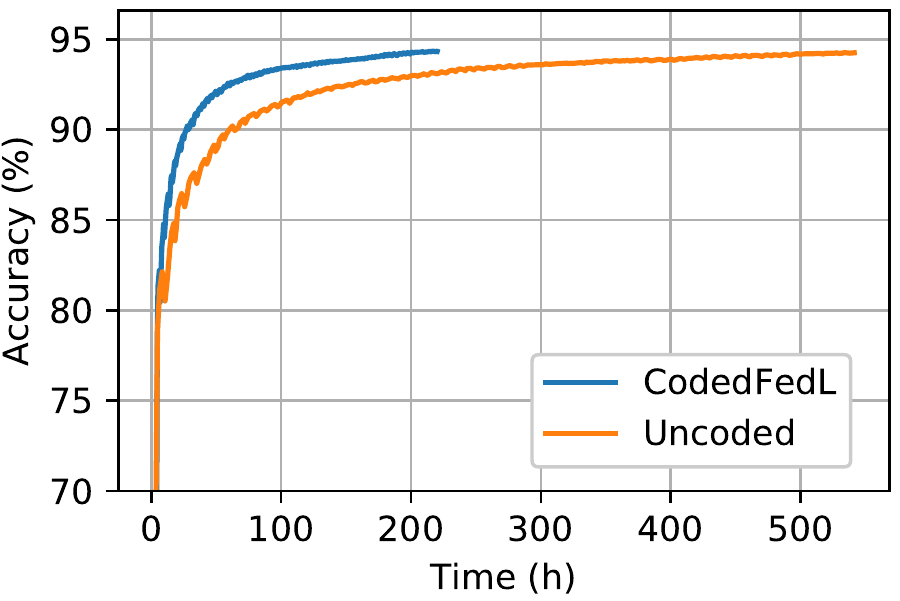}\label{fig:mnist}}\quad
  \subfigure[The test accuracy with respect to mini-batch update iteration under uncoded scheme vs. CodedFedL scheme.]{\includegraphics[scale=0.43]{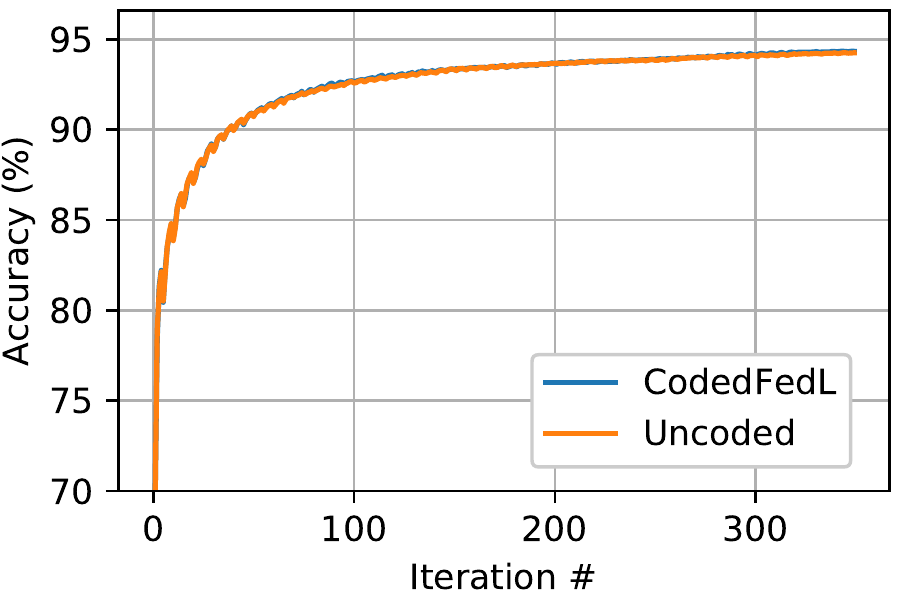}\label{fig:mnistiter}}\caption{Illustrating the results for MNIST.}\label{fig:results_mnist}
 \end{figure}
 \iffalse
 \begin{table}[htb!]
\caption{Summary of Results}\label{tab:summary}
\centering
\begin{tabular}{l |c c c c}%{|l|c|c|c|c|}
%\hline
Dataset & $\gamma$ (\%) & $t_\gamma^U$ (h) & $t_\gamma^{C}$ (h) & Gain \\ \hline
MNIST & 94.2 & 505 & 187 & $\times 2.70$\\ \hline 
\end{tabular}
\end{table}
\fi
\subsection{Simulation Setting}
We consider a wireless scenario consisting of $N{=}30$ heterogeneous client nodes and an MEC server. We consider both MNIST~\cite{lecun2010mnist} and Fashion-MNIST~\cite{xiao2017fashion}. Training is done in batches, and encoding for \name is applied based on the global mini-batch, with a coded redundancy of only $10\%$. The details of simulation parameters such as those for compute and communication resources, pre-processing steps, and strategy for modeling heterogeneity of data have been provided in Appendix \ref{sec:append2}. For each dataset, training is performed on the training set, while accuracy is reported on the test set. 
\subsection{Results}
Figure~\ref{fig:mnist} illustrates the generalization accuracy as a function of wall-clock time for MNIST, while Figure~\ref{fig:mnistiter} illustrates generalization accuracy vs training iteration. Similar results for Fashion-MNIST are included in Appendix \ref{sec:append2}. Clearly, CodedFedL has significantly better convergence time than the uncoded approach, and as highlighted in Section \ref{sec:codfedagg}, the coded federated gradient aggregation approximates the uncoded gradient aggregation well for large datasets. To illustrate this further, let $\gamma$ be target accuracy for a given dataset, while $t_\gamma^U$ and $t_\gamma^C$ respectively be the first time instants to reach the $\gamma$ accuracy for uncoded and CodedFedL respectively. In Table 1 in Appendix \ref{sec:append2}, we summarize the results, demonstrating a gain of up to $2.7\times$ in convergence time for CodedFedL over the uncoded scheme, even for a small coding redundancy of $10\%$.

\iffalse
\section{Conclusion}
\label{sec:conclusion}
\input{6-conclusion.tex}
\fi
%\pagebreak
%\newpage
\bibliography{biblio}
\bibliographystyle{icml2020}

\appendix
\section{Supplementary Material}

This is the supplementary component of our submission.
\subsection{Proof of Theorem}
\label{sec:append1}
\begin{theorem*} 
%\label{thm:expcRetWorker}
For the compute and communication models described in Section \ref{subsec:delayModel}, let $0{\leq}\wtilde{\ell}_j{\leq}\ell_j$ be the number of data points processed by $j$-th client in each training epoch. Then, for a waiting time of $t$ at the server, the expectation of the return $R_j(t;\wtilde{\ell}_j){=}\wtilde{\ell}_j {\mathbbm{1}}_{\{T_j\leq t\}}$ satisfies the following: 
\begin{align}
    &\Expc({R_{j}(t;\wtilde{\ell}_j)})\nonumber\\
    &=\left\{
	\begin{array}{ll}
		\sum_{\nu=2}^{\nu_m}{U}\bigg(t-\frac{\wtilde{\ell}_j}{\mu_j} -\tau_j \nu\bigg)h_{\nu} f_{\nu}(t;\wtilde{\ell}_j)  & \mbox{if } \nu_m\geq 2   \\
		0 & \mbox{otherwise } \nonumber
	\end{array}
    \right.
\end{align}
where $U(\cdot)$ is the unit step function,\\$f_{\nu}(t;\wtilde{\ell}_j)=\wtilde{\ell}_j\bigg(1-e^{-\frac{\alpha_j\mu_j}{\wtilde{\ell}_j}(t-\frac{\wtilde{\ell}_j}{\mu_j}-\tau_j \nu)}\bigg)$,\\
$h_{\nu}= (\nu-1)(1-p_j)^2p_j^{\nu-2}$, \\
and $\nu_m\in\mathbb{Z}$ satisfies $t-\tau_j \nu_m>0,\, t-\tau_j (\nu_m+1)\leq 0$.
%\vspace{-1cm}
\end{theorem*}
\begin{proof}
	Using the computation and communication models presented in Section \ref{subsec:delayModel}, we have the following for the execution time for one epoch for $j$-th client:
	\begin{align}
	\label{eq:theo1}
	T^{(j)}&=T^{(j,1)}_{cmp}+T^{(j,2)}_{cmp}+T^{(j)}_{com-d}+T^{(j)}_{com-u}\nonumber\\
	&=\frac{\wtilde{\ell}_j}{\mu_j}+T^{(j,2)}_{cmp}+\tau_j N^{(j)}_{com},
	\end{align}
	where $N^{(j)}_{com}{\sim} \text{NB}(r=2,p=1-p_j)$ has negative binomial distribution while $T^{(j,2)}_{cmp}{\sim} \text{E}(\frac{\alpha_j\mu_j}{\wtilde{\ell}_j})$ exponential. We have used the fact that $T^{(j)}_{com-d}$ and $T^{(j)}_{com-u}$ are IID $\text{G}(p)$ geometric random variables and sum of $r$ IID $\text{G}(p)$ is $\text{NB}(r,p)$. Therefore, the probability distribution for $T^{(j)}$ is obtained as follows:
	\begin{align}
	%\label{eq:theo2}
	&\Prob(T^{(j)}\leq t)=\Prob\left(\frac{\wtilde{\ell}_j}{\mu_j}+T_{cmp}^{(j,2)}+\tau_j N_{com}^{(j)}\leq t\right)\nonumber\\
	&=\sum_{\nu=2}^{\infty}\Prob(N_{com}^{(j)}=\nu)\nonumber\\ &\,\,\,\,\,\,\,\,\,\,\,\,\,\,\,\,\,\,\,\,\,\,\,\,\,\,\,\cdot\Prob\left(T_{cmp}^{(j,2)}\leq t-\frac{\wtilde{\ell}_j}{\mu_j}-\tau_j \nu|N_{com}^{(j)}=\nu\right)\nonumber\\
	&\overset{(a)}=\sum_{\nu=2}^{\infty}\Prob(N_{com}^{(j)}=\nu)\,\Prob\left(T_{cmp}^{(j,2)}\leq t-\frac{\wtilde{\ell}_j}{\mu_j}-\tau_j \nu\right)\nonumber\\
	&\overset{(b)}=\sum_{\nu=2}^{\infty}{U}\left(t-\frac{\wtilde{\ell}_j}{\mu_j} -\tau_j \nu\right)(\nu-1)(1-p_j)^2 p_j^{\nu-2}\nonumber\\
	&\,\,\,\,\,\,\,\,\,\,\,\,\,\,\,\,\,\,\cdot\left(1-\exp\left({-\frac{\alpha_j\mu_j}{\wtilde{\ell}_j}\left(t-\frac{\wtilde{\ell}_j}{\mu_j}-\tau_j \nu\right)}\right)\right)\nonumber
	\end{align}
	where $(a)$ holds due to independence of $T_{cmp}^{(j,2)}$ and $N_{com}^{(j)}$, while in $(b)$, we have used $U(\cdot)$ to denote the unit step function with $U(x){=}1$ for $x{>}0$ and $U(x){=}0$ for $x{\leq} 0$. For a fixed $t$, $\Prob(T^{(j)}{\leq} t){=}0$ if $t{\leq} 2\tau_j$. For $t{>}2\tau_j$, let $\nu_m{\geq} 2$ satisfy the following criteria:
	\begin{align}
	\label{eq:theo3}
	(t-\tau_j \nu_m)>0,
	(t-\tau_j (\nu_m+1))\leq 0. 
	\end{align}
	Therefore, for $\nu>\nu_m$, the terms in $(b)$ are $0$. Finally, as $\Expc({R_{j}(t;\wtilde{\ell}_j)}){=}\wtilde{\ell}_j \Expc({\bf{1}}_{\{T_j\leq t\}}){=}\wtilde{\ell}_j\Prob(T^{(j)}{\leq} t)$, we arrive at the result of our Theorem.
	%\vspace{-1cm}
\end{proof}

\subsection{Experiments}
\label{sec:append2}
\textbf{Simulation Setting}: We consider a wireless scenario consisting of $N{=}30$ client nodes and an MEC server, with similar computation and communication models as used in \cite{dhakal2019coded}. Specifically, we use an LTE network, where each node is assigned 3 resource blocks. We use the same failure probability $p_j{=}0.1$ for all clients, capturing the typical practice in wireless to adapt transmission rate for a constant failure probability. To model heterogeneity, link capacities (normalized) are generated using $\{1, k_1, k_1^2, ...\}$ and a random permutation of them is assigned to the clients, the maximum communication rate being $216$ kbps. An overhead of 10\% is assumed and each scalar is represented by 32 bits. The processing powers (normalized) are generated using $\{1, k_2, k_2^2, ...\}$, the maximum MAC rate being $3.072 \times 10^6$ MAC/s. We fix $(k_1,k_2)=(0.95,0.8)$.

We consider two different datasets: MNIST~\cite{lecun2010mnist} and Fashion-MNIST~\cite{xiao2017fashion}. The features are vectorized, and the labels are one-hot encoded. For kernel embedding, the hyperparameters are $(\sigma,q)=(5,2000)$. To model non-IID data distribution, training data is sorted by class label, and divided into 30 equally sized shard, one for each worker.  Furthermore, training is done in batches, where each global batch is of size 12000, i.e. each epoch constitutes 5 global mini-batch steps. Similarly, encoding for \name is applied based on the global mini-batch, with a coded redundancy of $10\%$. For both approaches, an initial step size of $6$ is used with a step decay of $0.8$ at epochs $40$ and $65$. Regularization parameter is  $0.000009$. For each dataset, training is performed on the training set, while accuracy is reported on the test set. Features are normalized to $[0,1]$ before kernel embedding, and we utilize the RBFSampler($\cdot$) function in the scikit library of Python for RFFM. Model parameters are initialized to $0$.

\textbf{Speedup Results}
Let $\gamma$ be target accuracy for a given dataset, while $t_\gamma^U$ and $t_\gamma^C$ respectively be the first time instants to reach the $\gamma$ accuracy for uncoded and CodedFedL respectively. In Table 1, we summarize the results for the two datasets.
\begin{table}[h]
\caption{Summary of Results}\label{tab:summary}
\centering
\begin{tabular}{l |c c c c}%{|l|c|c|c|c|}
%\hline
Dataset & $\gamma$ (\%) & $t_\gamma^U$ (h) & $t_\gamma^{C}$ (h) & Gain \\ \hline
MNIST & 94.2 & 505 & 187 & $\times 2.70$\\ %\hline
Fashion-MNIST & 84.2 & 513 & 216 & $\times 2.37$\\\hline 
\end{tabular}
\end{table}

\textbf{Convergence Curves for Fashion-MNIST}\\

As in the case for MNIST, CodedFedL provides significant improvement in convergence performance over the uncoded scheme for Fashion-MNIST as well as shown in Fig. \ref{fig:fmnist}. Fig. \ref{fig:fmnistiter} illustrates that coded federated aggregation described in Section 3.5 provides a good approximation of the true gradient over the entire distributed dataset across the client devices, even with a small coding redundancy of $10\%$.

\begin{figure}[h!]
	\centering
	\subfigure[The test accuracy with respect to the wall-clock time under uncoded scheme vs. CodedFedL scheme.]{\includegraphics[scale=0.43]{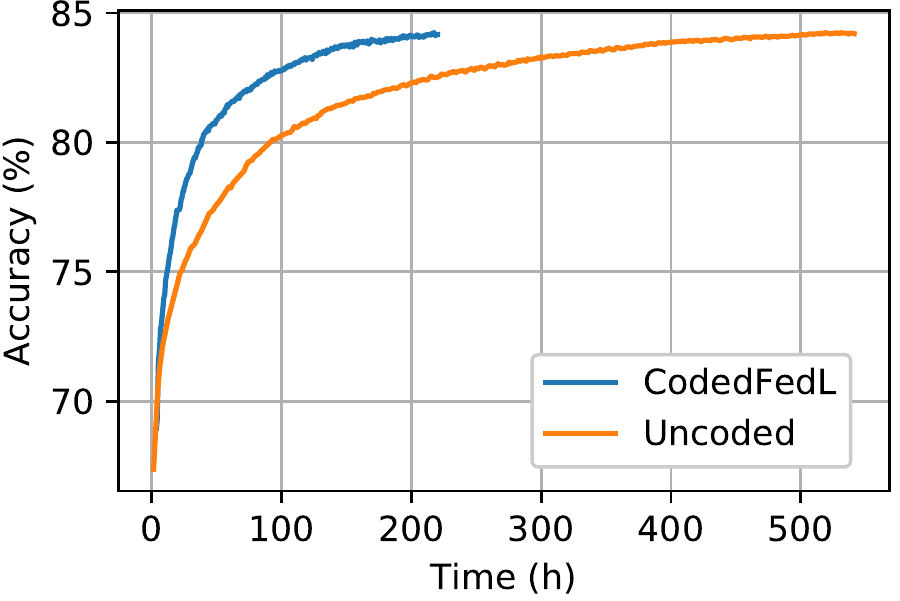}\label{fig:fmnist}}\quad
	\subfigure[The test accuracy with respect to mini-batch update iteration under uncoded scheme vs. CodedFedL scheme.]{\includegraphics[scale=0.43]{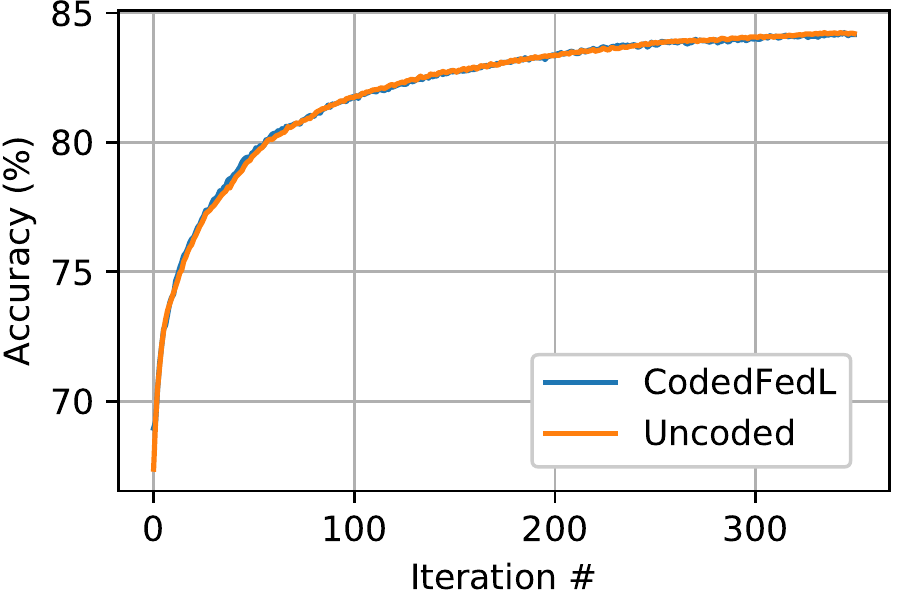}\label{fig:fmnistiter}}\caption{Illustrating the results for Fashion-MNIST.  }\label{fig:results_fmnist}
\end{figure}

\iffalse
\section{Supplementary Material}

\input{6-supplementary.tex}
\fi

\end{document}